\documentclass[aps,prl,reprint,superscriptaddress,amsmath,amssymb,longbibliography,floatfix]{revtex4-2}

\usepackage{bm}
\usepackage{mathtools}
\usepackage{microtype}
\usepackage{hyperref}
\hypersetup{hidelinks}
\usepackage{booktabs}
\usepackage{graphicx}
\usepackage{xcolor}
\usepackage{amsthm}

\newcommand{\R}{\mathbb{R}}
\newcommand{\C}{\mathbb{C}}

\newcommand{\dist}{\operatorname{dist}}
\newcommand{\articletitle}{A Fixed Universal Determinant is Variationally Complete for Continuum Fermions}

\newtheorem{theorem}{Theorem}
\newtheorem{lemma}[theorem]{Lemma}
\newtheorem{proposition}[theorem]{Proposition}

\theoremstyle{remark}
\newtheorem{remark}[theorem]{Remark}

\newcommand{\SN}{S_N}
\newcommand{\sgn}{\operatorname{sgn}}

\newcommand{\codim}{\operatorname{codim}}

\begin{document}

\title{\articletitle}

\author{Giuseppe Carleo}
\affiliation{Institute of Physics, \'Ecole Polytechnique F\'ed\'erale de Lausanne (EPFL), CH-1015 Lausanne, Switzerland}
\affiliation{Center for Quantum Science and Engineering, \'Ecole Polytechnique F\'ed\'erale de Lausanne (EPFL), CH-1015 Lausanne, Switzerland}

\author{Riccardo Rossi}
\affiliation{Institute of Physics, \'Ecole Polytechnique F\'ed\'erale de Lausanne (EPFL), CH-1015 Lausanne, Switzerland}
\affiliation{Center for Quantum Science and Engineering, \'Ecole Polytechnique F\'ed\'erale de Lausanne (EPFL), CH-1015 Lausanne, Switzerland}
\affiliation{CNRS, Laboratoire de Physique Th\'eorique de la Mati\`ere Condens\'ee, Sorbonne Universit\'e, 75005 Paris, France}
\date{\today}

\begin{abstract}
How many Slater determinants does an accurate variational description of interacting fermions require? Exact expansions in a finite basis need combinatorially many, and state-of-the-art fermionic neural quantum states stack growing numbers of them. We prove that, in the norms that govern variational calculations, at most two are needed, independently of the number of particles and of the target accuracy. A single universal Slater determinant---specified in advance, independent of both the system and the state---multiplied by a smooth bosonic wave function approximates any fermionic wave function in up to three spatial dimensions in the first-order Sobolev norm, which controls the variational energy. Reaching the second-order Sobolev norm---for Coulomb interactions, the domain of the Hamiltonian, which bounds the variance of the local energy at the core of variational Monte Carlo---requires at most one additional fixed determinant, and only in three dimensions. Antisymmetry therefore costs at most two universal determinants and no expressiveness: generalized Slater--Jastrow neural quantum states are variationally complete.
\end{abstract}

\maketitle

\textit{Neural quantum states and expressivity.}---Interacting fermions, from electrons in molecules and solids to dense nuclear matter, are the central and hardest instance of the quantum many-body problem: the Pauli principle forces the wave function to change sign under particle exchange, and building this antisymmetry into an expressive variational family has been a structural challenge of many-body physics for decades, from Slater--Jastrow and backflow trial states~\cite{FeynmanCohen1956,McMillan1965,CeperleyChesterKalos1977,KwonCeperley1998} to the nodal surfaces that control fixed-node quantum Monte Carlo~\cite{Reynolds1982,Ceperley1991,Foulkes2001}. Neural quantum states (NQS) parametrize the amplitude of each many-body configuration directly with a network~\cite{Carleo2017,Carleo2019}, and a substantial representation theory has since developed: bounds on the entanglement an architecture can support, constructions of volume-law and tensor-network states, and universality results relating network depth to expressive power~\cite{Gao2017,Deng2017,Levine2019,Sharir2022}. On a lattice, antisymmetry is naturally carried by the second-quantized fermionic algebra through Jordan--Wigner strings or a reference determinant in an occupation basis~\cite{Nomura2017,Choo2020}. In the continuum, the state of $N$ spinless fermions in $d$ spatial dimensions is instead a first-quantized antisymmetric wave function on $\R^{d\times N}$. In NQS, the antisymmetry is usually supplied by combining a permutation-equivariant network with explicit antisymmetric objects, typically Slater determinants or Pfaffians~\cite{Luo2019,Pfau2020,Hermann2020,vonGlehn2023,Cassella2023,Lovato2022,Pescia2022,HermannRev2023,kim_neural-network_2024,Moreno2022,Pescia2024}.

How many antisymmetric blocks are then unavoidable? The answer depends sharply on the notion of approximation and on what is allowed as a block (Table~\ref{tab:hierarchy}). Exact expansions in a finite basis require combinatorially many determinants, and a single generalized backflow determinant can represent any fermionic wave function exactly only at the price of orbitals that are discontinuous for $d>1$ and target dependent~\cite{Hutter2020}. For uniform pointwise approximation with continuous factors the cost remains steep: $K=O(\epsilon^{-Nd})$ backflow--Vandermonde terms at accuracy $\epsilon$~\cite{Han2019}, and $K=dN+1$ fixed determinants for an exact continuous representation~\cite{Ye2024}, a count reduced by Fu's Fermi Sets to $K=1$ in $d=1$, $K=2$ in $d=2$, and $K\le dN+1$ for $d\ge3$~\cite{Fu2026}. The obstruction is geometric: a fixed determinant carries an artificial nodal surface that no continuous bosonic factor can cure uniformly.
\begin{table*}[t]
\caption{Representative universality bounds for fermionic states.
Here $K$ denotes the number of determinants.}
\label{tab:hierarchy}
\begin{ruledtabular}
\begin{tabular}{@{}llll@{}}
Work & Approximation & Antisymmetric representation & Bound on $K$ \\
\hline
Han \emph{et al.}~\cite{Han2019}
& uniform 
& backflow--Vandermonde sum
& $O(\epsilon^{-Nd})$ \\

Ye \emph{et al.}~\cite{Ye2024}
& exact 
& fixed determinants $\times\;C^0$
& $dN+1$ \\

Fu~\cite{Fu2026}
& uniform
& fixed determinants $\times \;C^0_+$ 
& $1$ ($d=1$), $2$ ($d=2$), $dN+1$ ($d\ge3$) \\

\hline
\textbf
{This work}
& {$H^s(\mathbb{R}^{d\times N})$}
&fixed determinants $\times \;C^\infty_{c,+}$
& $\bm 1\, (H^1, d\le 3),\;\bm 1\, (H^2, d\le 2),\;  \bm 2\, (H^2,d=3)$ \\
\end{tabular}
\end{ruledtabular}
\end{table*}

Uniform approximation is, however, different than what physics requires: the supremum norm constrains the amplitude at every configuration, while energies and observables are integrals over particle positions involving as integrand the value of the wave function and of the derivatives. More natural settings for variational calculations are instead Sobolev spaces. For a Coulomb Hamiltonian in three dimensions, the quadratic form of the energy is defined on the first-order Sobolev space $H^1(\mathbb{R}^{d\times N})$~\cite{Lieb1983}, and the operator domain is the second-order space $H^2(\mathbb{R}^{d\times N})$~\cite{Kato1951}, whose norm controls the variance of the local energy, the central estimator of variational Monte Carlo (VMC)~\cite{McMillan1965,CeperleyChesterKalos1977,Foulkes2001}. It is in these topologies that the fermionic variational representation question should be posed.

In this Letter we prove that, in the Sobolev norms, the cost of antisymmetry collapses. For $d\leq3$, a single fixed, universal complex Slater determinant multiplied by a smooth compactly-supported bosonic wave function is dense in the fermionic $H^1$ space; in $H^2$, one determinant suffices for $d\leq2$ and at most two for $d=3$ (Table~\ref{tab:hierarchy}). These counts are independent of the particle number and of the target accuracy: a growing determinant expansion is unnecessary for variational completeness, and the entire expressive burden can be carried by one---at the $H^2$ level, at most two---bosonic wave functions. For neural quantum states, the construction amounts to a generalized Slater--Jastrow wave function whose Jastrow factor is sign indefinite and not restricted to few-body correlations.

\textit{Wave function spaces.} For simplicity of discussion, we restrict ourselves here to spinless particles. The extension to spinful particles is straightforward and discussed later. We consider $N$ particles in $d$ spatial dimensions, described by square-integrable wave functions $\psi(\bm r_1,\dots,\bm r_N)$ on $\mathbb{R}^{d\times N}$, forming the space $L^2(\mathbb{R}^{d\times N})$. Particle statistics selects a subspace, indicated by a $\pm$ subscript: fermionic wave functions are antisymmetric under particle permutations and form $L^2_-$, bosonic ones are symmetric and form $L^2_+$. The continuous-space Hamiltonian is $H =
\sum_{i=1}^{N}
\left[
-\frac{\hbar^{2}}{2m}\nabla_i^{2}
+ V_{\mathrm{ext}}(\mathbf{r}_i)
\right]
+ \frac{1}{2}\sum_{i\neq j}
U(\mathbf{r}_i-\mathbf{r}_j)$, where $U$ denotes the two-body Coulomb interaction and $V_{\text{ext}}$ a one-body potential. Wave functions of finite average kinetic energy form the first-order Sobolev space $H^1_{-}(\mathbb{R}^{d\times N})$ which, in three dimensions and under mild assumptions on $V_{\text{ext}}$, is also the space of finite average energy~\cite{Lieb1983}. The domain $D(H)$ on which the Hamiltonian itself is defined coincides in three dimensions with the second-order Sobolev space, for fermions $H^2_{-}(\mathbb{R}^{d\times N})$: the functions of $H^1_{-}$ with square-integrable Laplacian~\cite{Kato1951}. The Hamiltonian domain is not a mere technicality in VMC practice: the main Monte Carlo estimator is the local energy 
\begin{equation}
    E_{\text{loc}}(\bm r_1,\dots,\bm r_N)=\frac{(H\psi)(\bm r_1,\dots,\bm r_N)}{\psi(\bm r_1,\dots,\bm r_N)},
\end{equation}
and $\psi\in H^2$ guarantees that its Monte Carlo variance is bounded. $H^1$ and $H^2$ are the two Sobolev norms in which we pose, and answer, the fermionic representation question.

For neural-network architectures optimized with gradient methods, working with smooth functions subject to no antisymmetry constraint is a practical necessity. We denote by $C^\infty_c$ the set of compactly-supported smooth functions. We will use the well-known fact that $C^\infty_c(\mathbb{R}^{d\times N})$ is dense in $H^s(\mathbb{R}^{d\times N})$ in the $H^s$ norm~\cite{AronszajnSmith1961}. We can now state the main result of this work. 

\begin{theorem}\label{thm:main}
Let $s\in\{1,2\}$ and $d\le3$, and let $\psi\in H^s_-(\mathbb{R}^{d\times N})$ be a spinless fermionic wave function of $N$ particles in $d$ spatial dimensions. Then there exist sequences of smooth compactly-supported bosonic $N$-particle wave functions $\phi_{\alpha,n}\in C^\infty_{c,+}(\mathbb{R}^{d\times N})$ such that
\begin{equation}
\lim_{n\to\infty} \left\lVert\psi - \sum_{\alpha=1}^{K_s^{(d)}}\mathcal{D}_{\alpha}^{(d)} \,\phi_{\alpha,n}\right\rVert_{H^s}=0,
\end{equation}
where the $\mathcal{D}_{\alpha}^{(d)}$ are universal, $\psi$-independent, antisymmetric polynomials in the particle coordinates, and where $K_1^{(d)}=1$ for $d\in\{1,2,3\}$, $K_2^{(d)} = 1$ for $d\in\{1,2\}$, and $K_2^{(3)}\le 2$.
\end{theorem}

The complete proof of Theorem~\ref{thm:main} is presented in the Supplemental Material, while below we sketch the main conceptual steps. After introducing the main objects of the proof, we first discuss the $H^1$ case and then show the modifications needed to establish the $H^2$ variant. Finally, we note at the end that the theorem is also applicable to the spinful case without modification.
\begin{figure*}[t]
\centering
\includegraphics[width=0.9\textwidth]{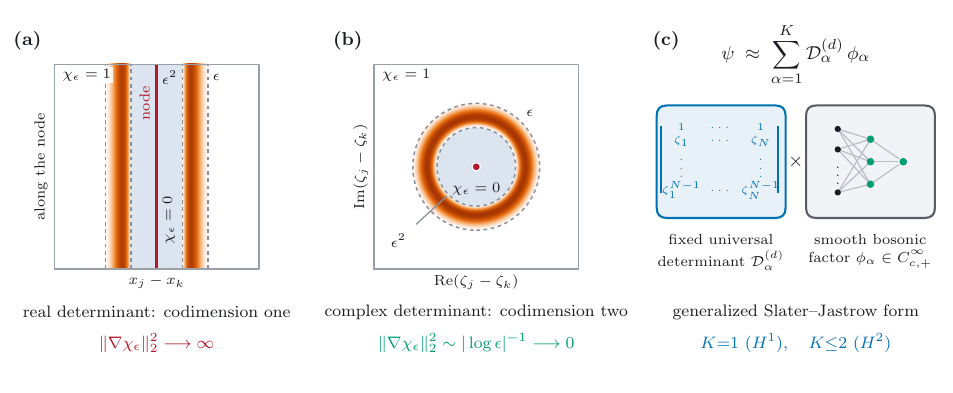}
\caption{\textit{Why one fixed complex determinant is variationally complete in $d>1$, and generalized Slater-Jastrow representation.} In panels (a) and (b), the heat maps show the calculated kinetic-energy cost $|\nabla\chi_\epsilon|^2$ for a particular choice of $\chi_\epsilon$, and the shaded region is the deleted set where $\chi_\epsilon=0$. (a) A real determinant (see Eq.~\eqref{eq-real-det}) has a codimension-one nodal surface: erasing it forces the cutoff $\chi_\epsilon$ across a shrinking strip, at divergent kinetic energy cost $\|\nabla\chi_\epsilon\|_2^2\to\infty$. (b) A complex determinant (see Eq.~\eqref{eq-cpx-det}) imposes two real equations, so the nodal surface has codimension two and admits the insertion of the cutoff at vanishing kinetic energy cost $\|\nabla\chi_\epsilon\|_2^2\sim|\log\epsilon|^{-1}\to0$.  (c) The resulting representation from Theorem~\ref{thm:main} is a generalized Slater--Jastrow form; $\psi\underset{H^1}{\approx}\mathcal{D}_1^{(d)}\phi_1$ and $\psi\underset{H^2}{\approx}\sum_{\alpha=1}^K\mathcal{D}_\alpha^{(d)}\phi_\alpha$ with fixed universal determinants times smooth (optionally compactly-supported) bosonic factors $\phi_\alpha\in C^\infty_{c,+}$, with $K=1$ in $H^1$ ($d\le3$) and $K\le2$ in $H^2$. This schematic was produced with \textsc{Matplotlib} code written by Claude Fable~5 (Anthropic).}
\label{fig:capacity}
\end{figure*}

\textit{Construction of the universal determinants.} We give below the explicit expression of the fixed Slater determinants $\mathcal{D}_\alpha^{(d)}$, which are all of the Vandermonde form. Let $\mathcal{N}_{\alpha}^{(d)}$ be the set of zeros of $\mathcal{D}_{\alpha}^{(d)}$, which we refer to as the nodal surface. We also introduce the collision nodal surface $\mathcal{N}_c^{(d)}$, which is the set of points where two or more particles share the same position and where any continuous fermionic wave function must vanish.

 For $d=1$ it is sufficient to consider the real determinant
\begin{equation}\label{eq-real-det}
    \mathcal{D}_{1}^{(1)}(\bm r_1,\dots,\bm r_N) = \prod_{j<k} (x_k-x_j) = \det\big(x_j^{\,k-1}\big)_{j,k=1}^{N}.
\end{equation}
  The theorem is particularly simple to motivate in one dimension using the $\mathcal{D}_{1}^{(1)}$ antisymmetric factor: as $\mathcal{N}_{1}^{(1)}\equiv \mathcal{N}_c^{(1)}$, any smooth antisymmetric function $\psi_n\in C^\infty_{c,-}$ approximating $\psi\in H^s_-$ can  be safely divided by $\mathcal{D}_{1}^{(1)}$, providing another smooth symmetric function $\phi_{1,n}=\frac{\psi_n}{\mathcal{D}_1^{(1)}}\in C^\infty_{c,+}$. This result is perhaps not surprising in view of the long history of fermion-boson mappings in one dimension~\cite{Girardeau1960}.

For $d>1$, we introduce the complex determinant
\begin{equation}\label{eq-cpx-det}
    \mathcal{D}_{1}^{(d)}(\bm r_1,\dots,\bm r_N) = \prod_{j<k} (\zeta_k-\zeta_j) = \det\big(\zeta_j^{\,k-1}\big)
\end{equation}
where $\zeta_j = x_j+iy_j\in\mathbb{C}$, and for $d=3$ we define an additional complex determinant
\begin{equation}
    \mathcal{D}_{2}^{(3)}(\bm r_1,\dots,\bm r_N)=\prod_{j<k}(\xi_k-\xi_j)=\det\big(\xi_j^{\,k-1}\big),
\end{equation}
where $\xi_j = y_j + i z_j\in\mathbb{C}$.  These forms of Slater determinants were also considered in previous works~\cite{Ye2024,Fu2026}.

\textit{The $\mathcal{N}_{1}^{(d)}$ nodal surface.} We focus here on $\mathcal{N}_1^{(d)}$ for $d>1$, which is the union of $\mathcal{N}_{1;(j,k)}^{(d)}$ for all pairs of particles $(j,k)$
\begin{equation}
    \mathcal{N}_{1;(j,k)}^{(d)} =\left\{(\bm r_1,\dots,\bm r_N)\in \mathbb{R}^{d\times N}\,\Big|\, x_j=x_k,\; y_j=y_k\right\},
\end{equation}
which is a vector subspace of codimension $q=2$.

\textit{Introducing the cutoff function for $\mathcal{N}_1^{(d)}$.} The central device of the proof is a smooth cutoff that deletes the artificial nodal surface at vanishing Sobolev cost. We introduce $\chi_\epsilon\in C^\infty$, identically zero near $\mathcal{N}_1^{(d)}$ and converging to $1$ everywhere expect the nodal surface as $\epsilon\to 0$. Let $\rho(\bm r_1,\dots,\bm r_N)$ be a smooth version of the minimal distance of $(\bm r_1,\dots,\bm r_N)$ from $\mathcal{N}_1^{(d)}$. We choose the cutoff function $\chi_\epsilon$ to be a function of $\rho$, and to be equal to $1$ for $\rho \ge \epsilon$. We remark that, as the nodal surface $\mathcal{N}_1^{(d)}$ is symmetric,  the cutoff function can be chosen symmetric, $\chi_{\epsilon}\in C^\infty_+$. Loosely speaking, $\chi_\epsilon$ must not have a large gradient, as this might spoil the $H^1$ convergence. Crucially, it turns out that one can choose $\chi_\epsilon$ that satisfies the bounds
    $|\chi_\epsilon'(\rho)| \le \frac{C}{\log \epsilon^{-1} \;\rho}$,  $|\chi_\epsilon''(\rho)|\le \frac{C}{\log \epsilon^{-1}\,\rho^2}$, and such that $\chi_\epsilon(\rho)=0$ for $\rho\le \epsilon^2$. These bounds follow directly from the choice $\chi_\epsilon(\rho)=\theta(\log \tau_\epsilon)$, where $\tau_\epsilon=\left(\frac{\rho}{\epsilon^2}\right)^{\frac{1}{\log \epsilon^{-1}}}$, and where $\theta$ is a smooth function that goes from zero to one in $[0,1]$.

\textit{The deletion lemma for the $H^1$ case.}---The essential lemma we use in the proof states that if $u\in C^\infty_c$, then $u\chi_\epsilon\to u$ in $H^1$. To prove it, let us first remark that $u-u\chi_\epsilon$ is clearly in $L^2$, while the gradient reads
\begin{equation}
    ||\nabla(u-u\chi_\epsilon)||_2\le ||\nabla\,u\;(1-\chi_\epsilon)||_2+||u||_{\infty}\, ||\nabla\chi_\epsilon||_{L^2(S)},
\end{equation}
where $S$ is the support of $u$.
The first term vanishes as $\epsilon\to 0$ by dominated convergence, as $1-\chi_\epsilon\to0$ almost everywhere and is bounded by one. The only potentially-problematic term therefore comes from the norm of the gradient of the cutoff function on the compact support of $u$,
\begin{equation}
     ||\nabla\chi_\epsilon||_{L^2(S)}^2 \sim \int_{\epsilon^2}^\epsilon d\rho\, \rho^{q-1} |\chi'_\epsilon|^2\le  \frac{C}{(\log \epsilon^{-1})^2}\int_{\epsilon^2}^\epsilon d\rho\, \rho^{q-3} ,
\end{equation}
where the factor $\rho^{q-1}$ comes from the Jacobian for the change of variable to the $\rho$ coordinate, and where $q$ is the codimension of the subspace $\mathcal{N}_{1;(j,k)}^{(d)}$ (see Fig.~\ref{fig:capacity}b for a sketch). The norm of the gradient of $\chi_\epsilon$ therefore goes to zero for $\epsilon\to 0$ for $q\ge 2$. We also see why using $\mathcal{D}_1^{(1)}$ evaluated on a single Cartesian coordinate would not have been enough for $d>1$: the corresponding nodal surface is made of codimension-one subspaces (see Fig.~\ref{fig:capacity}(a) for a sketch).

\textit{Assembling the proof for the $H^1$ case.}---Let $\psi_n\in C^\infty_{c,-}$ be a sequence converging to $\psi\in H^1_-$ with $||\psi-\psi_n||_{H^1}\le \frac{1}{n}$. Let $\psi_{n,m}=\chi_{\frac{1}{m}}\, \psi_n\in C^\infty_{c,-}$. By the deletion lemma, $\psi_{n,m}\to \psi_n$ in $H^1$ when $m\to \infty$. Let $m(n)$ be such that $||\psi_n-\psi_{n,m(n)}||_{H^1}\le \frac{1}{n}$. Define
 \begin{equation}
     \phi_{1,n} = \frac{\psi_{n,m(n)}}{\mathcal{D}_{1}^{(d)}}\in C^\infty_{c,+}.
 \end{equation}
Then, $||\psi - \mathcal{D}_{1}^{(d)}\,\phi_{1,n}||_{H^1}\le \frac{2}{n}$, which shows that $\psi$ can be approximated in the first-order Sobolev norm by a  generalized Slater-Jastrow wave function (see Fig.~\ref{fig:capacity}(c)).

\textit{Extension to the $H^2$ case.}---It is possible to extend the deletion lemma above to show convergence of $u\chi_\epsilon$ to $u$ in $H^2$ if the minimal codimension of the subspaces of the nodal surface is $q\ge 4$. $\mathcal{N}_1^{(d)}$ is made of subspaces of codimension $q=2$, so the lemma is not directly applicable. 

We first consider $d=2$. In this case, $\mathcal{N}_{1}^{(2)}$ coincides with the collision nodal surface $\mathcal{N}_c^{(2)}$, which means that the wave function must be zero there. If a smooth function is zero on the nodal surface, a variant of the deletion lemma in $H^2$ decreases the minimal codimension to $q\ge 2$, which concludes the proof of the two-dimensional case. 

In the three-dimensional case, we are led to consider the two complex determinants $\mathcal{D}_{1}^{(3)}$ and $\mathcal{D}_{2}^{(3)}$. We focus on the intersection of the nodal surfaces $\mathcal{N}_{1}^{(3)}$ and $\mathcal{N}_{2}^{(3)}$ of the two determinants, $\mathcal{N}^{(3)}_{12} =\cup_{e,e'\in\mathcal{E}}\,L_{e,e'}$, where $ \mathcal{E}$ is the set of pairs of particles, and $L_{(j,k),(j',k')}=\{(\bm r_1,\dots,\bm r_N)\in\mathbb{R}^{d\times N}|\zeta_j=\zeta_k,\;\xi_{j'}=\xi_{k'}\}$ (see Fig.~\ref{fig:mechanism} for a sketch). $L_{e,e}$ has codimension $q=3$, and it is contained in the collision nodal surface, which means that we can apply the variant of the deletion lemma with minimal codimension $q=2$. $L_{e,e'}$ with $e\neq e'$ has codimension $q=4$. We can therefore conclude that $u\chi_\epsilon\to u$ in $H^2$ for every $u\in C^\infty_c$ and for a smooth $\chi_\epsilon$ that is zero near $\mathcal{N}_{12}^{(3)}$. For $\psi_n\in C^\infty_{c,-}$ converging towards $ \psi\in H^2_-$, we define, similarly to what was done for the $H^1$ case
\begin{equation}
    \phi_{\alpha,n} = \chi_{\frac{1}{m(n)}} \;\frac{[\mathcal{D}_\alpha^{(3)}]^* \,\psi_n}{|\mathcal{D}_1^{(3)}|^2+|\mathcal{D}_2^{(3)}|^2}\in C^\infty_{c,+},
\end{equation}
which is a smooth function as it is identically zero by construction near the problematic set $\mathcal{D}_1^{(3)}=\mathcal{D}_2^{(3)}=0$. This concludes the presentation of the proof of Theorem~\ref{thm:main}.

\textit{Extension to the spinful case.} We consider a fixed number $N_\sigma$ of spin-$\sigma$ particles, for $\sigma\in \{\uparrow,\downarrow\}$. The fermionic wave function is antisymmetric in the coordinates for the spin-up and spin-down particle separately. Accordingly, it is sufficient to define the determinants as the products of determinants for the spin-up and the spin-down particles separately
\begin{equation}
    \mathcal{D}_\alpha^{(d)}(\bm R_\uparrow,\bm R_{\downarrow})=    \mathcal{D}_\alpha^{(d)}(\bm R_\uparrow)    \;\mathcal{D}_\alpha^{(d)}(\bm R_{\downarrow}),
\end{equation}
where $\bm R_\sigma =(\bm r_{1,\sigma},\dots,\bm r_{N_\sigma,\sigma})$, and the determinant count $K_s^{(d)}$ is identical to the spinless case.

\textit{Discussion.}---Theorem~\ref{thm:main} provides a representation-theoretic baseline for continuum fermionic variational calculations: at the level of the variational energy, a single fixed Vandermonde determinant multiplied by a smooth bosonic factor already spans the whole fermionic $H^1$ space; controlling the local energy---hence the Monte Carlo variance---requires at most one additional determinant, and only in three dimensions. The number of explicitly antisymmetric objects is therefore not a source of expressiveness for neural quantum states: a single-determinant generalized Slater--Jastrow ansatz with a sufficiently flexible bosonic network is variationally complete, and additional determinants or Pfaffian factors should be understood as improving the conditioning of the bosonic factor near the artificial nodal surfaces rather than enlarging the reachable state space. We note however that the theorem is silent on optimization: the exact factors $\phi_{\alpha,n}$ must reconstruct the true nodal structure of the wave function, and quantifying the difficulty of this task---together with approximation rates as a function of $N$---remains an open problem.

\begin{figure}[t]
\centering
\includegraphics[width=0.99\columnwidth]{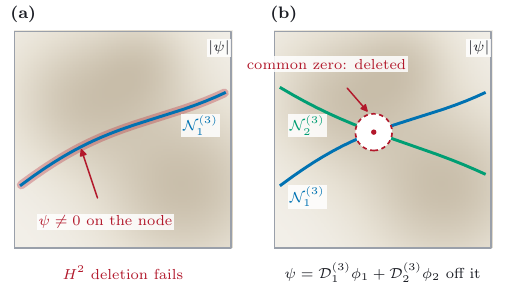}
\caption{\textit{Why two fixed determinants suffice in $H^2$ for $d=3$}. The background shade intensity represents $|\psi|$ for a generic smooth target in a two-dimensional slice of configuration space. (a) A single complex determinant $\mathcal{D}_1^{(3)}$ has an \emph{off-collision} nodal set of codimension two on which a generic $\psi\in H^2$ does not vanish, so the deletion lemma cannot be applied. 
(b) The nodal surface of a  second complex determinant $\mathcal{D}_2^{(3)}$, $\mathcal{N}_2^{(3)}$, shares with $\mathcal{N}_1^{(3)}$ only a codimension-three collision nodal surface and a codimension-four set; both are deleted by the  cutoff at vanishing $H^2$ cost (dashed disk); we can the safely approximate $\psi\underset{H^2}{\approx}\mathcal{D}_1^{(3)}\phi_1+\mathcal{D}_2^{(3)}\phi_2$ with smooth bosonic $\phi_\alpha$. This schematic was produced with \textsc{Matplotlib} code written by Claude Fable~5 (Anthropic).}
\label{fig:mechanism}
\end{figure}

\begin{acknowledgments}
The proof of the Theorem presented in this work was developed by ChatGPT~5.6 (Pro and Sol) and Claude Opus~4.8 after extensive interactions with the Authors, and part of the manuscript was drafted and edited with help from ChatGPT~5.6, Claude Opus~4.8, and Claude Fable~5. The Authors have curated the proof and the presentation, and carefully checked every result. Accordingly, the Authors declare full responsibility on the presented material. This research was supported by SEFRI through Grant No. MB22.00051 (NEQS - Neural Quantum Simulation).
\end{acknowledgments}

\bibliography{one_determinant}

\onecolumngrid
\vspace{0.5em}
\begin{center}
\rule{0.45\textwidth}{0.4pt}\\[0.4em]
{\large\textbf{End Matter}}\\[-0.2em]
\rule{0.45\textwidth}{0.4pt}
\end{center}
\vspace{0.3em}
\twocolumngrid

\setcounter{equation}{0}
\renewcommand{\theequation}{EM\arabic{equation}}

\textit{Use of AI tools.}---The constructions and their proofs---the single-determinant $H^1$ completeness argument, the logarithmic-deletion proof of the fixed two-determinant $H^2$ theorem variant, and the supporting codimension---were developed through extended technical exchange with ChatGPT~5.6 (Pro and Sol) and Claude Opus~4.8. The same tools implemented the numerical cross-checks of both theorems (the capacity-cost, incidence, and logarithmic-deletion tests) and ran a multi-agent adversarial audit that re-derived the load-bearing steps and independently verified the bibliography; the schematic figures were generated with \textsc{Matplotlib} code written by Claude Fable~5, as disclosed in the figure captions. They, together with Claude Fable~5, also assisted in drafting and revising the text. The Authors checked and curated every definition, statement, proof, and numerical result, and take full responsibility for the content of the paper.

\onecolumngrid

\clearpage

\begin{center}

{\large\bfseries
Supplemental Material for
``\articletitle''
\par}

\vspace{1.2em}

{\normalsize
Giuseppe Carleo$^{1,2}$ and Riccardo Rossi$^{1,2,3}$
\par}

\vspace{0.7em}

{\small
$^{1}$Institute of Physics, \'Ecole Polytechnique F\'ed\'erale de Lausanne (EPFL), CH-1015 Lausanne, Switzerland\\
$^{2}$Center for Quantum Science and Engineering, \'Ecole Polytechnique F\'ed\'erale de Lausanne (EPFL), CH-1015 Lausanne, Switzerland\\
$^3$CNRS, Laboratoire de Physique Th\'eorique de la Mati\`ere Condens\'ee, Sorbonne Universit\'e, 75005 Paris, France
\par}

\vspace{0.5em}

\end{center}

\vspace{1.5em}

\setcounter{secnumdepth}{2}

\makeatletter

\renewcommand{\section}{
  \@startsection{section}{1}{\z@}
    {2.2ex \@plus .8ex \@minus .2ex}
    {0.9ex \@plus .2ex}
    {\normalfont\normalsize\bfseries\raggedright}
}

\renewcommand{\subsection}{
  \@startsection{subsection}{2}{\z@}
    {1.8ex \@plus .6ex \@minus .2ex}
    {0.7ex \@plus .2ex}
    {\normalfont\small\bfseries\raggedright}
}

\let\@hangfrom@section\@hang@from
\let\@hangfrom@subsection\@hang@from

\renewcommand{\@seccntformat}[1]{
  \csname the#1\endcsname.\hspace{0.4em}
}

\makeatother

\setcounter{theorem}{0}

\setcounter{section}{0}
\renewcommand{\thesection}{S\arabic{section}}

\setcounter{subsection}{0}
\renewcommand{\thesubsection}{S\arabic{section}.\arabic{subsection}}

\setcounter{equation}{0}
\renewcommand{\theequation}{S\arabic{equation}}

\setcounter{figure}{0}
\renewcommand{\thefigure}{S\arabic{figure}}

\setcounter{table}{0}
\renewcommand{\thetable}{S\arabic{table}}

This document contains the proof of the following Theorem, which was presented in the main text of the manuscript, whose statement we reproduce below.

\begin{theorem}\label{thm:main}
Let $s\in\{1,2\}$. If $\psi\in H^s_-(\mathbb{R}^{d\times N})$ is a spinless fermionic wave function of $N$ particles in $d$ spatial dimensions, it exists a finite set of sequences of smooth compactly-supported bosonic $N$-particle wave functions $\phi_{\alpha,n}\in C^\infty_{c,+}(\mathbb{R}^{d\times N})$ such that
\begin{equation}
\lim_{n\to\infty} \left\lVert\psi - \sum_{\alpha=1}^{K_s^{(d)}}\mathcal{D}_{\alpha}^{(d)} \,\phi_{\alpha,n}\right\rVert_{H^s}=0,
\end{equation}
where $\mathcal{D}_{\alpha}^{(d)}$ are universal, $\psi$-independent, antisymmetric polynomials in the particle coordinates. We show that $K_1^{(d)}=1$ for $d\in\{1,2,3\}$, $K_2^{(d)} = 1$ for $d\in\{1,2\}$, and $K_2^{(3)}\le 2$.
\end{theorem}

\section{Statement and notation}

Set
\[
  \mathcal X_{d,N}:=(\R^d)^N\cong\R^{dN},
  \qquad
  X=(x_1,\ldots,x_N).
\]
For $\sigma\in\SN$, define
\[
  \sigma X=(x_{\sigma^{-1}(1)},\ldots,x_{\sigma^{-1}(N)}).
\]
The symmetric and antisymmetric Sobolev spaces are
\begin{align*}
 H^s_+(\mathcal X_{d,N})
 &=\{\phi\in H^s:\phi(\sigma X)=\phi(X)\text{ for every }\sigma\in\SN\},\\
 H^s_-(\mathcal X_{d,N})
 &=\{\psi\in H^s:\psi(\sigma X)=\sgn(\sigma)\psi(X)\text{ for every }\sigma\in\SN\}.
\end{align*}
All functions are complex-valued, as is natural for Schroedinger wave functions.

For a real-linear map $p:\R^d\to\C$, define the projected Vandermonde
\begin{equation}\label{eq:projected-vandermonde}
  \Delta_p(X)
  :=\prod_{1\leq i<j\leq N}\bigl(p(x_i)-p(x_j)\bigr).
\end{equation}
Since a permutation only reorders the factors with the usual Vandermonde sign,
\begin{equation}\label{eq:antisymmetry-delta}
  \Delta_p(\sigma X)=\sgn(\sigma)\Delta_p(X).
\end{equation}

We use the following fixed choices.

\begin{itemize}
\item If $d=1$,
\[
  \mathcal{D}_1^{(1)}(X)=\Delta(X):=\prod_{i<j}(x_i-x_j).
\]
\item If $d\geq 2$,
\[
  p_1(x)=x^{(1)}+i x^{(2)},
  \qquad \mathcal{D}_1^{(d)}=\Delta_{p_1}.
\]
\item If $d=3$,
\[
  p_2(x)=x^{(2)}+i x^{(3)},
  \qquad \mathcal{D}_2^{(3)}=\Delta_{p_2}.
\]
\end{itemize}

The proof of the theorem is straightforward once the deletion lemma is proven in Sec.~\ref{sec:lemma}.

\section{Reduction to smooth compactly-supported fermionic functions}

Define the antisymmetrizer
\begin{equation}\label{eq:antisymmetrizer}
  (P_-u)(X)
  :=\frac1{N!}\sum_{\sigma\in\SN}
  \sgn(\sigma)u(\sigma^{-1}X).
\end{equation}
Every permutation acts isometrically on $H^s(\mathcal X_{d,N})$, hence
$P_-$ is a bounded projection on $H^s$.

\begin{lemma}\label{lem:smooth-core}
For every integer $s\geq0$,
\[
 C_{c,-}^\infty(\mathcal X_{d,N})
 :=C_c^\infty(\mathcal X_{d,N})\cap H^s_-(\mathcal X_{d,N})
\]
is dense in $H^s_-(\mathcal X_{d,N})$.
\end{lemma}

\begin{proof}
Take $\psi\in H^s_-$.  Since $C_c^\infty(\mathcal X_{d,N})$ is dense in
$H^s(\mathcal X_{d,N})$, choose $u_n\in C_c^\infty$ with
$u_n\to\psi$ in $H^s$.  Then $P_-u_n\in C_{c,-}^\infty$ and
\[
 \|P_-u_n-\psi\|_{H^s}
 =\|P_-(u_n-\psi)\|_{H^s}
 \leq \|P_-\|\,\|u_n-\psi\|_{H^s}\longrightarrow0,
\]
because $P_-\psi=\psi$.
\end{proof}

It is therefore enough to start with a fixed
$f\in C_{c,-}^\infty$ and approximate this $f$ by the required factorizations.

\section{The logarithmic cutoff lemma}\label{sec:lemma}

The nodal sets of the projected Vandermonde determinants are finite unions of linear subspaces.
At the critical Sobolev codimension, to be introduced below, an ordinary cutoff is too expensive. We show instead that a logarithmic cutoff has vanishing Sobolev cost.

\subsection{A regularized distance to a finite union of subspaces}

Let
\[
  Z=\bigcup_{\nu=1}^J L_\nu\subset\R^M,
\]
where each $L_\nu$ is a linear subspace.  Put
\[
  \delta_\nu(X):=\dist(X,L_\nu),
  \qquad
  \delta(X):=\dist(X,Z)=\min_\nu\delta_\nu(X).
\]
For $X\notin Z$, define
\begin{equation}\label{eq:regularized-distance}
  \rho(X)
  :=\left(\sum_{\nu=1}^J\delta_\nu(X)^{-2}\right)^{-1/2}.
\end{equation}
Then
\begin{equation}\label{eq:rho-comparable}
  \frac{\delta(X)}{\sqrt J}\leq \rho(X)\leq\delta(X).
\end{equation}
Indeed, if $\delta=\min_\nu\delta_\nu$, then
\[
  \delta^{-2}
  \leq\sum_\nu\delta_\nu^{-2}
  \leq J\delta^{-2},
\]
and taking inverse square roots gives \eqref{eq:rho-comparable}.

Let $P_\nu$ be the orthogonal projection onto $L_\nu^\perp$.  Then
$\delta_\nu(X)^2=|P_\nu X|^2$.  Differentiating
\eqref{eq:regularized-distance} gives the pointwise estimates
\begin{equation}\label{eq:rho-derivatives}
  |D\rho(X)|\leq C,
  \qquad
  |D^2\rho(X)|\leq \frac{C}{\rho(X)},
  \qquad X\notin Z.
\end{equation}
For completeness, write
\[
 R(X):=\sum_\nu\delta_\nu(X)^{-2},
 \qquad \rho=R^{-1/2}.
\]
Since
\[
 |D(\delta_\nu^{-2})|\leq C\delta_\nu^{-3},
 \qquad
 |D^2(\delta_\nu^{-2})|\leq C\delta_\nu^{-4},
\]
and $\delta_\nu\geq\rho$, one obtains
\[
 |DR|\leq C\rho^{-3},
 \qquad
 |D^2R|\leq C\rho^{-4}.
\]
Now
\[
 D\rho=-\frac12R^{-3/2}DR,
\]
and
\[
 D^2\rho
 =\frac34R^{-5/2}DR\otimes DR
 -\frac12R^{-3/2}D^2R.
\]
Using $R=\rho^{-2}$ gives \eqref{eq:rho-derivatives}.

Suppose that every $L_\nu$ has codimension at least $q$.  On every fixed compact set
$K\subset\R^M$, the tubular-neighborhood estimate
\begin{equation}\label{eq:tube-volume}
  \bigl|K\cap\{\rho<t\}\bigr|\leq C_K t^q,
  \qquad 0<t<1,
\end{equation}
holds.  To see this, use \eqref{eq:rho-comparable}:
\[
 \{\rho<t\}\subset\{\delta<\sqrt Jt\}
 \subset\bigcup_{\nu=1}^J\{\delta_\nu<\sqrt Jt\}.
\]
A tube of radius $t$ around a codimension-$q_\nu$ linear subspace has volume
$O(t^{q_\nu})$ inside a fixed compact set, and $q_\nu\geq q$.

A consequence used repeatedly below is the weighted estimate
\begin{equation}\label{eq:weighted-integral}
 I_a(\varepsilon)
 :=\int_{K\cap\{\varepsilon^2<\rho<\varepsilon\}}\rho^{-a}\,dX
 \leq
 \begin{cases}
 C\varepsilon^{q-a},&a<q,\\
 C|\log\varepsilon|,&a=q.
 \end{cases}
\end{equation}
One derivation is by dyadic shells.  On
\[
 2^{-k-1}\varepsilon<\rho\leq2^{-k}\varepsilon
\]
we have $\rho^{-a}\leq C(2^k/\varepsilon)^a$, whereas
\eqref{eq:tube-volume} bounds the shell volume by
$C(2^{-k}\varepsilon)^q$.  Thus the $k$th contribution is bounded by
\[
 C\varepsilon^{q-a}2^{-k(q-a)}.
\]
The geometric series converges when $a<q$.  When $a=q$, each shell contributes at most a constant, and the number of shells between $\varepsilon^2$ and $\varepsilon$ is
$O(|\log\varepsilon|)$.

\subsection{The cutoff and its derivatives}

Fix $\theta\in C^\infty(\R)$ such that
\[
 0\leq\theta\leq1,
 \qquad
 \theta(t)=0\text{ for }t\leq0,
 \qquad
 \theta(t)=1\text{ for }t\geq1.
\]
For $0<\varepsilon<e^{-1}$ and $L_\varepsilon:=|\log\varepsilon|$, define
\begin{equation}\label{eq:eta-epsilon}
 \eta_\varepsilon(t)
 :=\theta\left(\frac{\log(t/\varepsilon^2)}{L_\varepsilon}\right),
 \qquad t>0,
\end{equation}
and set
\begin{equation}\label{eq:chi-epsilon}
 \chi_\varepsilon(X):=\eta_\varepsilon(\rho(X)).
\end{equation}
Then
\[
 \chi_\varepsilon=0\quad\text{when }\rho\leq\varepsilon^2,
 \qquad
 \chi_\varepsilon=1\quad\text{when }\rho\geq\varepsilon.
\]
Although $\rho$ need not be smooth on $Z$, the function $\chi_\varepsilon$ is globally smooth, because it is identically zero in a neighborhood of $Z$.

Differentiating \eqref{eq:eta-epsilon} gives, on the transition region,
\[
 |\eta_\varepsilon'(t)|\leq\frac{C}{L_\varepsilon t},
 \qquad
 |\eta_\varepsilon''(t)|\leq\frac{C}{L_\varepsilon t^2}.
\]
Combining these estimates with \eqref{eq:rho-derivatives} yields
\begin{equation}\label{eq:chi-derivatives}
 |\nabla\chi_\varepsilon|
 \leq\frac{C}{L_\varepsilon\rho},
 \qquad
 |\nabla\otimes\nabla\chi_\varepsilon|
 \leq\frac{C}{L_\varepsilon\rho^2}.
\end{equation}

\begin{lemma}[Logarithmic deletion lemma]\label{lem:deletion}
Let $Z$ be a finite union of linear subspaces of $\R^M$, each of codimension at least $q$, and let $u\in C_c^\infty(\R^M)$.

\begin{enumerate}
\item If $q\geq2$, then
\[
 u\chi_\varepsilon\longrightarrow u\quad\text{in }H^1(\R^M).
\]
\item If $q\geq4$, then
\[
 u\chi_\varepsilon\longrightarrow u\quad\text{in }H^2(\R^M).
\]
\item If $q\geq2$ and $u=0$ on $Z$, then
\[
 u\chi_\varepsilon\longrightarrow u\quad\text{in }H^2(\R^M).
\]
\end{enumerate}
If $Z$ is invariant under particle permutations, then $\chi_\varepsilon$ may be chosen symmetric.
\end{lemma}

\begin{proof}
Let
\[
 w_\varepsilon:=u-u\chi_\varepsilon=u(1-\chi_\varepsilon).
\]
All terms without derivatives of $\chi_\varepsilon$ are supported in
$\{\rho<\varepsilon\}$ and converge to zero in $L^2$ by
\eqref{eq:tube-volume}.  We only estimate the terms containing derivatives of
$\chi_\varepsilon$.

For $H^1$,
\[
 \nabla w_\varepsilon
 =(1-\chi_\varepsilon)Du-u\nabla\chi_\varepsilon.
\]
Using \eqref{eq:chi-derivatives} and \eqref{eq:weighted-integral},
\[
 \|u\nabla\chi_\varepsilon\|_{L^2}^2
 \leq\frac{C}{L_\varepsilon^2}
 \int_{\{\varepsilon^2<\rho<\varepsilon\}}\rho^{-2}\,dX.
\]
If $q=2$, the right-hand side is $O(L_\varepsilon^{-1})$.
If $q>2$, it is $O(\varepsilon^{q-2}L_\varepsilon^{-2})$.
This proves part (i).

For $H^2$,
\begin{equation}\label{eq:second-derivative-product}
 \nabla\otimes\nabla w_\varepsilon
 =(1-\chi_\varepsilon)\nabla\otimes\nabla u
 -2\,\operatorname{sym}(\nabla u\otimes \nabla\chi_\varepsilon)
 -u\nabla\otimes\nabla\chi_\varepsilon.
\end{equation}
The middle term satisfies
\[
 \|\nabla u\otimes \nabla\chi_\varepsilon\|_{L^2}^2
 \leq\frac{C}{L_\varepsilon^2}I_2(\varepsilon).
\]
If $q\geq4$, this tends to zero.  The last term satisfies
\[
 \|u\nabla\otimes\nabla\chi_\varepsilon\|_{L^2}^2
 \leq\frac{C}{L_\varepsilon^2}I_4(\varepsilon).
\]
For $q=4$, this is $O(L_\varepsilon^{-1})$; for $q>4$, it is
$O(\varepsilon^{q-4}L_\varepsilon^{-2})$.  This proves part (ii).

Now suppose $u=0$ on $Z$.  Since $u$ is smooth and compactly supported, it is Lipschitz.  Therefore
\[
 |u(X)|\leq C\dist(X,Z)\leq C\sqrt J\,\rho(X),
\]
where the last inequality follows from \eqref{eq:rho-comparable}.  Consequently,
\[
 |u\nabla\otimes\nabla\chi_\varepsilon|
 \leq\frac{C}{L_\varepsilon\rho},
\]
and hence
\[
 \|u\nabla\otimes\nabla\chi_\varepsilon\|_{L^2}^2
 \leq\frac{C}{L_\varepsilon^2}I_2(\varepsilon).
\]
The same bound holds for the middle term in
\eqref{eq:second-derivative-product}.  It tends to zero for every $q\geq2$,
with the critical value $q=2$ giving $O(L_\varepsilon^{-1})$.
This proves part (iii).

Finally, if the collection $\{L_\nu\}$ is permutation invariant, then the sum in
\eqref{eq:regularized-distance} is permutation invariant.  Thus $\rho$ and
$\chi_\varepsilon$ are symmetric.
\end{proof}

\section{The one-dimensional case: exact smooth divisibility}

Assume $d=1$ and set
\[
  \Delta(X)=\prod_{i<j}(x_i-x_j).
\]
In this case the factorization is exact on the smooth fermionic core.

\begin{proposition}\label{prop:d1}
For every $\Psi\in C_{c,-}^\infty(\R^N)$, there exists
$\Phi\in C_{c,+}^\infty(\R^N)$ such that
\[
  \Psi=\Delta \Phi.
\]
\end{proposition}

\begin{proof}
Fix a pair $i<j$ and introduce
\[
 R=\frac{x_i+x_j}{2},
 \qquad r=x_i-x_j.
\]
Keeping all other variables fixed, antisymmetry under the transposition $(ij)$ gives
\[
 \Psi(R,r)=-\Psi(R,-r).
\]
In particular, $\Psi(R,0)=0$.  The fundamental theorem of calculus gives the smooth
Hadamard factorization
\begin{equation}\label{eq:hadamard}
 \Psi(R,r)
 =r\int_0^1 \partial_r \Psi(R,tr)\,dt.
\end{equation}
Thus $\Psi$ is smoothly divisible by $x_i-x_j$.

Apply this argument successively to all pair hyperplanes.  Dividing by one factor does not destroy vanishing on another pair hyperplane: away from the intersection of the two hyperplanes this is immediate, and continuity extends the vanishing to the intersection.  Therefore
\[
 \Phi:=\frac{\Psi}{\prod_{i<j}(x_i-x_j)}
\]
extends to a globally smooth function.

Away from the collision set,
\[
 \Phi(\sigma X)
 =\frac{\Psi(\sigma X)}{\Delta(\sigma X)}
 =\frac{\sgn(\sigma)\Psi(X)}{\sgn(\sigma)\Delta(X)}
 =\Phi(X).
\]
By continuity, $\Phi$ is symmetric everywhere.  Since $\Psi$ is compactly supported and
$\Phi$ vanishes outside the same compact set, $\Phi\in C_{c,+}^\infty$.
\end{proof}

This proves Theorem~\ref{thm:main} for $d=1$ in both $H^1$ and $H^2$.

\section{One factor in $H^1$ for every $d\geq2$}

Let
\[
 p_1(x)=x^{(1)}+ix^{(2)},
 \qquad
 \mathcal{D}_1^{(d)}(X)=\prod_{i<j}\bigl(p_1(x_i)-p_1(x_j)\bigr).
\]
Its zero set is
\begin{equation}\label{eq:Z1}
 Z_1
 =\bigcup_{i<j}L_{ij}^{(1)},
 \qquad
 L_{ij}^{(1)}:=\{X:p_1(x_i-x_j)=0\}.
\end{equation}
Each $L_{ij}^{(1)}$ is a real linear subspace of codimension two, because
$p_1(x_i-x_j)=0$ is the pair of real equations
\[
 x_i^{(1)}-x_j^{(1)}=0,
 \qquad
 x_i^{(2)}-x_j^{(2)}=0.
\]
The union $Z_1$ is permutation invariant.

Take $\Psi\in C_{c,-}^\infty$.  By Lemma~\ref{lem:deletion}(i), there are symmetric
cutoffs $\chi_n$ such that
\[
 \Psi_n:=\Psi\chi_n\longrightarrow \Psi\quad\text{in }H^1,
 \qquad
 \Psi_n=0\quad\text{in a neighborhood of }Z_1.
\]
Since $\Psi$ is antisymmetric and $\chi_n$ is symmetric, $\Psi_n$ is antisymmetric.
Define
\begin{equation}\label{eq:b-one-factor}
 \Phi_n(X)=
 \begin{cases}
 \Psi_n(X)/\mathcal{D}_1^{(d)}(X),&X\notin Z_1,\\
 0,&X\text{ in the neighborhood of }Z_1\text{ where }\Psi_n=0.
 \end{cases}
\end{equation}
This is globally smooth, because it is identically zero near every zero of $\mathcal{D}_1^{(d)}$.
It is compactly supported.  Moreover,
\[
 \Phi_n(\sigma X)
 =\frac{\sgn(\sigma)\Psi_n(X)}{\sgn(\sigma)\mathcal{D}_1^{(d)}(X)}
 =\Phi_n(X),
\]
so $\Phi_n$ is symmetric.  Finally,
\[
 \Psi_n=\mathcal{D}_1^{(d)}\Phi_n,
 \qquad
 \mathcal{D}_1^{(d)}\Phi_n\longrightarrow \Psi\quad\text{in }H^1.
\]
This proves the $H^1$ assertion.

\begin{remark}
The use of the complex projection $x^{(1)}+ix^{(2)}$ is essential.  A real one-coordinate Vandermonde has a codimension-one nodal set in $d>1$, and such a set cannot be deleted at arbitrarily small $H^1$ cost.
\end{remark}

\section{One factor in $H^2$ when $d=2$}

Now let $d=2$.  The map
\[
 p_1(x)=x^{(1)}+ix^{(2)}
\]
is injective as a real-linear map.  Therefore
\[
 p_1(x_i-x_j)=0\quad\Longleftrightarrow\quad x_i=x_j.
\]
Hence the zero set of $\mathcal{D}_1^{(2)}$ is exactly the collision set
\begin{equation}\label{eq:collision-set-d2}
 \mathcal C
 :=\bigcup_{i<j}\{X:x_i=x_j\}.
\end{equation}
Each pair-collision subspace has codimension $d=2$.

If $\Psi$ is antisymmetric, then $\Psi$ vanishes on $\mathcal C$.  Indeed, at a point with
$x_i=x_j$, the transposition $(ij)$ leaves the configuration unchanged, while
antisymmetry gives
\[
 \Psi(X)=\Psi((ij)X)=-\Psi(X),
\]
so $\Psi(X)=0$.

Lemma~\ref{lem:deletion}(iii), with $q=2$, gives symmetric cutoffs $\chi_n$ such that
\[
 \Psi_n:=\Psi\chi_n\longrightarrow \Psi\quad\text{in }H^2,
 \qquad
 \Psi_n=0\quad\text{near }\mathcal C.
\]
The quotient $\Phi_n=\Psi_n/\mathcal{D}_1^{(2)}$, extended by zero near $\mathcal C$, is therefore in
$C_{c,+}^\infty$, exactly as in \eqref{eq:b-one-factor}.  Thus
\[
 \Psi_n=\mathcal{D}_1^{(2)}\Phi_n\longrightarrow \Psi\quad\text{in }H^2.
\]

The critical estimate can be seen directly.  Near a collision, write
$r=x_i-x_j\in\R^2$.  Since $\Psi=0$ at $r=0$ and $\Psi$ is smooth,
$|\Psi|\leq C|r|$.  In the logarithmic annulus
$\varepsilon^2<|r|<\varepsilon$,
\[
 |D\chi_\varepsilon|\lesssim\frac1{|r||\log\varepsilon|},
 \qquad
 |D^2\chi_\varepsilon|\lesssim\frac1{|r|^2|\log\varepsilon|}.
\]
Consequently,
\[
 |\nabla\Psi\,\nabla\chi_\varepsilon|^2
 +|\Psi D^2\chi_\varepsilon|^2
 \lesssim
 \frac1{|r|^2|\log\varepsilon|^2}.
\]
Integration in the two normal variables gives
\[
 \int_{\varepsilon^2}^{\varepsilon}
 \frac{r\,dr}{r^2|\log\varepsilon|^2}
 =\frac1{|\log\varepsilon|},
\]
which tends to zero.

\section{Two factors in $H^2$ when $d\geq3$}

Let
\[
 \mathcal{D}_1^{(3)}=\Delta_{p_1},
 \qquad
 \mathcal{D}_2^{(3)}=\Delta_{p_2},
\]
with $p_1,p_2$ chosen as in Section~1.  Define their common zero set
\[
 Z_{12}:=\{X:\mathcal{D}_1^{(3)}(X)=\mathcal{D}_2^{(3)}(X)=0\}.
\]
If $\mathcal E$ denotes the set of unordered particle pairs, then
\begin{equation}\label{eq:common-zero-union}
 Z_{12}
 =\bigcup_{e,f\in\mathcal E}L_{e,f},
 \qquad
 L_{e,f}:=\{X:p_1(r_e)=0,\ p_2(r_f)=0\},
\end{equation}
where $r_{ij}=x_i-x_j$.

\subsection{Codimension count in $d=3$}

For $d=3$,
\[
 p_1(r)=r^{(1)}+ir^{(2)},
 \qquad
 p_2(r)=r^{(2)}+ir^{(3)}.
\]
If $e=f$, the four displayed real equations reduce to
\[
 r_e^{(1)}=r_e^{(2)}=r_e^{(3)}=0.
\]
Thus
\[
 L_{e,e}=\{x_i=x_j\}
\]
is a genuine collision subspace of codimension three.

If $e\neq f$, the four real equations are independent and
$\codim L_{e,f}=4$.  Here is a direct verification.  Let $v_e\in\R^N$ be the incidence vector of the edge $e$: it has one entry $+1$, one entry $-1$, and all other entries zero.  The four linear forms are
\[
 v_e\otimes e_1^*,
 \quad
 v_e\otimes e_2^*,
 \quad
 v_f\otimes e_2^*,
 \quad
 v_f\otimes e_3^*.
\]
In a linear dependence, the first-coordinate block forces the coefficient of
$v_e\otimes e_1^*$ to vanish, and the third-coordinate block forces the coefficient of
$v_f\otimes e_3^*$ to vanish.  The remaining relation is
\[
 a v_e+b v_f=0.
\]
Distinct unordered edges have non-collinear incidence vectors, hence $a=b=0$.

Write
\[
 Z_3:=\bigcup_eL_{e,e},
 \qquad
 Z_4:=\bigcup_{e\neq f}L_{e,f}.
\]
The function $\Psi$ vanishes on $Z_3$, because $Z_3$ is a union of collision sets.
Lemma~\ref{lem:deletion}(iii), with $q=3$, gives antisymmetric smooth functions
$g_n$ such that
\[
 g_n\to \Psi\quad\text{in }H^2,
 \qquad
 g_n=0\quad\text{near }Z_3.
\]
For each fixed $n$, apply Lemma~\ref{lem:deletion}(ii) to the codimension-four union
$Z_4$.  Choose a second cutoff so that the resulting $\Psi_n$ satisfies
\[
 \|\Psi_n-g_n\|_{H^2}<\frac1n,
 \qquad
 \Psi_n=0\quad\text{near }Z_3\cup Z_4=Z_{12}.
\]
Then $\Psi_n\to \Psi$ in $H^2$.

\subsection{Codimension count in $d\geq4$}

For $d\geq4$,
\[
 p_1(r)=r^{(1)}+ir^{(2)},
 \qquad
 p_2(r)=r^{(3)}+ir^{(4)}.
\]
The two real equations from $p_1(r_e)=0$ involve only coordinate directions $1,2$, whereas those from $p_2(r_f)=0$ involve only directions $3,4$.
Therefore all four equations are independent, for every $e,f$, including $e=f$.
Thus every $L_{e,f}$ in \eqref{eq:common-zero-union} has codimension four.
Lemma~\ref{lem:deletion}(ii) gives antisymmetric $\Psi_n\in C_c^\infty$ such that
\[
 \Psi_n\to \Psi\quad\text{in }H^2,
 \qquad
 \Psi_n=0\quad\text{near }Z_{12}.
\]

\subsection{Algebraic decomposition away from the common zero set}

In both cases, $d=3$ and $d\geq4$, we have constructed antisymmetric
$\Psi_n\in C_c^\infty$ such that
\[
 \Psi_n\to \Psi\quad\text{in }H^2,
 \qquad
 \Psi_n=0\quad\text{near }Z_{12}.
\]
Set
\begin{equation}\label{eq:S-denominator}
 S(X):=|\mathcal{D}_1^{(3)}(X)|^2+|\mathcal{D}_2^{(3)}(X)|^2.
\end{equation}
On $\mathcal X_{d,N}\setminus Z_{12}$, define
\begin{equation}\label{eq:two-coefficients}
 \phi_{\alpha,n}(X)
 :=\Psi_n(X)\frac{\overline{\mathcal{D}_\alpha^{(d)}(X)}}{S(X)},
 \qquad \alpha=1,2.
\end{equation}
Since $\Psi_n$ vanishes in a neighborhood of $Z_{12}$, each coefficient extends by zero to a globally smooth compactly supported function.

The denominator $S$ is symmetric.  Since both $\Psi_n$ and $\mathcal{D}_\alpha^{(d)}$ are antisymmetric,
\[
 \phi_{\alpha,n}(\sigma X)
 =\frac{\sgn(\sigma)\Psi_n(X)\,
        \overline{\sgn(\sigma)\mathcal{D}_\alpha^{(d)}(X)}}{S(X)}
 =\phi_{\alpha,n}(X).
\]
Thus $\phi_{\alpha,n}\in C_{c,+}^\infty$.  Finally,
\begin{align*}
 \mathcal{D}_1^{(3)}\Phi_{1,n}+\mathcal{D}_2^{(3)}\Phi_{2,n}
 &=\Psi_n\frac{|\mathcal{D}_1^{(3)}|^2+|\mathcal{D}_2^{(3)}|^2}{S}\\
 &=\Psi_n.
\end{align*}
Therefore
\[
 \mathcal{D}_1^{(3)}\Phi_{1,n}+\mathcal{D}_2^{(3)}\Phi_{2,n}\longrightarrow \Psi
 \quad\text{in }H^2.
\]
This proves the constructive part of Theorem~\ref{thm:main} for $d\geq3$.

\section{Passage from the smooth core to all of $H^s_-$}

Let $\psi\in H^s_-$ and choose $\Psi_m\in C_{c,-}^\infty$ with
\[
 \|\psi-\Psi_m\|_{H^s}<\frac1m
\]
by Lemma~\ref{lem:smooth-core}.  For each $m$, the preceding sections give a function
\[
 G_m=\sum_{\alpha=1}^{K_s^{(d)}}\mathcal{D}_\alpha^{(d)} \Phi_{\alpha,m},
 \qquad \Phi_{\alpha,m}\in C_{c,+}^\infty,
\]
such that
\[
 \|\Psi_m-G_m\|_{H^s}<\frac1m.
\]
Hence
\[
 \|\psi-G_m\|_{H^s}
 \leq\|\psi-\Psi_m\|_{H^s}+\|\Psi_m-G_m\|_{H^s}
 <\frac2m\longrightarrow0.
\]
$\square$.

\end{document}